\documentclass[letterpaper, 10 pt, conference]{ieeeconf}
\IEEEoverridecommandlockouts
\usepackage{graphicx}
\usepackage{hyperref}
\usepackage{amsmath,amssymb}
\usepackage{flushend}
\usepackage{color}
\usepackage{latexsym}\usepackage{comment}
\usepackage{url}
\usepackage{graphics}
\usepackage{pdfsync}
\usepackage{stfloats}
\usepackage{epsfig}
\usepackage{times}
\usepackage{amssymb}
\usepackage{amsmath}
\usepackage{amsfonts}
\usepackage{textcomp}
\usepackage{xcolor}
\usepackage{indentfirst}
\usepackage{epstopdf}
\usepackage{multirow}
\usepackage{float}
\usepackage{color}
\usepackage{enumerate}
\usepackage{subfigure}
\usepackage{comment}
\usepackage{cite}
\usepackage{caption}
\usepackage{bm}

\usepackage{tikz}
\usetikzlibrary{calc,positioning,decorations.pathreplacing}

\usepackage{algorithm}
\usepackage{algpseudocode}

\newcommand{\ba}{\begin{array}}
	\newcommand{\ea}{\end{array}}
\newcommand{\be}{\begin{equation}}
\newcommand{\ee}{\end{equation}}
\newcommand{\bea}{\begin{eqnarray}}
\newcommand{\eea}{\end{eqnarray}}
\newcommand{\bean}{\begin{eqnarray*}}
	\newcommand{\eean}{\end{eqnarray*}}
\newcommand{\bc}{\begin{center}}
	\newcommand{\ec}{\end{center}}

\title{\LARGE \bf Convergence and Robustness of Value and Policy Iteration for the Linear Quadratic Regulator}
\author{Bowen~Song, Chenxuan~Wu, Andrea~Iannelli
	\thanks{Bowen Song acknowledges the support of the International Max Planck Research School for Intelligent Systems (IMPRS-IS). Andrea Iannelli acknowledges the German Research Foundation (DFG) for support of this work under Germany’s Excellence Strategy - EXC 2075 – 390740016.
	}
	\thanks{The authors are with the University of Stuttgart, Institute for Systems Theory and Automatic Control, 70550 Stuttgart, Germany
		{\tt\small  \{bowen.song, andrea.iannelli\}@ist.uni-stuttgart.de},
  {\tt\small  \{st176873@stud.uni-stuttgart.de\}}. }
}

\newtheorem{Theorem}{Theorem}
\newtheorem{Definition}{Definition}
\newtheorem{Lemma}{Lemma}
\newtheorem{Remark}{Remark}
\newtheorem{Corollary}{Corollary}

\begin{document}
	
	\maketitle

\begin{abstract}
This paper revisits and extends the convergence and robustness properties of value and policy iteration algorithms for discrete-time linear quadratic regulator problems. In the model-based case, we extend current results concerning the region of exponential convergence of both algorithms. In the case where there is uncertainty on the value of the system matrices, we provide input-to-state stability results capturing the effect of model parameter uncertainties. Our findings offer new insights into these algorithms at the heart of several approximate dynamic programming schemes, highlighting their convergence and robustness behaviors. Numerical examples illustrate the significance of some of the theoretical results.
\end{abstract}

\section{Introduction}
Approximate dynamic programming (ADP) \cite{7927616,bertsekas2022abstract,lewis2013reinforcement} is a powerful algorithmic approach designed to solve sequential decision-making problems across a broad range of applications. Two fundamental approaches in ADP are: value iteration (VI) and policy iteration (PI), which have been extensively analyzed in the literature \cite{bertsekas2019reinforcement,7362040}. VI updates the value function of the underlying optimal control problem iteratively \cite{doi:10.1137/S0363012995291609,BERTSEKAS2020100003}, while PI evaluates and improves policies sequentially \cite{doi:10.1137/S0363012902399824}. Convergence properties of these two algorithms have been studied in works such as \cite{9790821} and \cite{SemiCone}. In \cite{9790821}, the convergence rates of VI and PI are compared for a finite state and action Markov decision problem. In contrast, \cite{SemiCone} investigates the conditions for asymptotic and exponential convergence of VI and PI in the context of the discrete-time linear quadratic regulator (LQR) problem. ADP's application to the LQR problem has received substantial attention due to its analytically tractable nature, making it an ideal benchmark for studying ADP in environments with continuous state and action spaces \cite{pmlr-v119-park20c,NEURIPS2019_aaebdb8b}. Studying the performance of VI and PI for the LQR problem is an active area of research \cite{SemiCone,9511623,9762997,FAN2024105731}. 

Typically, performing VI or PI requires knowledge of the system model, which is where most existing theoretical results are established \cite{SemiCone,1099755,pmlr-v119-park20c}. However, the system model is often unavailable in practice. To address this, data can be used to either identify a model and based on the estimate apply the algorithm (indirect data-driven control) or directly design the controller (direct data-driven Control). In \cite{IJRNC}, PI is combined with online model estimation, while \cite{9691800} proposes a direct formulation of PI using online data, bypassing model estimation. For both approaches to VI and PI, analyzing the robustness of the algorithms is crucial due to unavoidable uncertainty associated with the use of finite and potentially noisy data, which can introduce estimation errors that affect the controller's design. In \cite{RobustPI,PANG2021240}, the robustness of PI applied to continuous-time LQR problems and stochastic LQR problems is analyzed, respectively. Inspired by \cite{RobustPI}, our previous work \cite{IJRNC} extended this analysis to the robustness of PI for discrete-time LQR problems.

In this work, we investigate the nominal (i.e. with known model) exponential convergence and robustness of VI and PI applied to the discrete-time LQR problem. For the nominal case, we extend the standard conditions for exponential convergence of VI and PI algorithms as provided in \cite{SemiCone}. Building on these results, we analyze the robustness to model uncertainties of VI and PI, 
an aspect that has been explored for PI under different uncertainty structures in \cite{RobustPI, PANG2021240}, but not for VI.
Specifically, we study the performance of two algorithms when estimates of the true model are used, and we analyze the effect of uncertainty on the convergence properties. A motivating example for this analysis is the use of online identification routines providing at each iteration of the PI/VI algorithm a different estimate, which differs from the true one by a certain amount. We show that both VI and PI algorithms have inherent robustness to uncertainties within specific bounds. This property is crucial for the reliable deployment of indirect VI and PI algorithms, where handling uncertainties and estimation errors should be considered. 

The paper is organized as follows. Section \ref{sec:P} introduces the problem setting and some preliminaries. Section \ref{sec:C} and Section \ref{Robust} detail the exponential convergence and robustness analysis for both VI and PI algorithms, respectively. Section \ref{sec:S} provides simulations to illustrate some theoretical examples. Section \ref{conclusion} concludes the work. 

\emph{Notations:} We denote by $A\succeq 0$ and $A\succ0$ a positive semidefinite and positive definite matrix $A$, respectively. For matrices, $\lVert \cdot\rVert_F$ and $\lVert \cdot\rVert$ denote respectively their Frobenius norm and induced $2$-norm. Given positive definite matrix $P_\epsilon \preceq I$\cite[Lemma 6]{SemiCone}, we define the induced $P_\epsilon$-norm for matrices of compatible dimension as: $\| \cdot \|_{P_\epsilon} := \sqrt{\lambda_{max}((\cdot)^{\top}P_\epsilon(\cdot))}$.  For $X \in \mathbb{R}^{m \times n}$, we define $vec(X) := [X_{1}^\top, ..., X_{n}^\top]^\top$, where $X_i$ is the $i$-th column of matrix $X$. Kronecker product is represented as $\otimes$. For $Y \in \mathbb{R}^{m \times n}$ and $r>0$, we define $\mathcal{B}_{r}(Y) := \{X \in \mathbb{R}^{m \times n} | \|X-Y\|_F < r\}$. 
A sequence $\{Y_i\}$ is a map $\mathbb{Z}_+\rightarrow \mathbb{R}^{n \times m}$. 
For bounded scalar sequences, we denote by $\lVert Y \rVert_\infty:=\sup\limits_{i\in{\mathbb{Z}_+}}\{Y_i\}$.

\section{Preliminaries and Problem Setting}\label{sec:P}
We consider discrete-time linear time-invariant systems of the form:
\begin{equation}\label{LTI}
  x_{t+1}=Ax_t+Bu_t,
\end{equation}
where $x_t\in \mathbb{R}^{n_x}$ is the system state, $u_t\in \mathbb{R}^{n_u}$ is the control input, $t$ is the timestep, and pair $(A,B)$ is stabilizable. The objective is to design a state-feedback controller $u_t=Kx_t$ that minimizes the following infinite horizon cost:
\begin{equation}\label{Cost}
  J(x_t,K)=\sum\limits_{k=t}^{+\infty} r(x_k,u_k)=\sum\limits_{k=t}^{+\infty} x_k^\top Qx_k+u_k^\top Ru_k,
\end{equation}
where $R\succ 0$ and $Q\succeq 0$. 
Given a linear state-feedback gain $K$ that is stabilizing (i.e. $A+BK$ is Schur stable), the corresponding cost $J(x_t, K)$  can be expressed as $x_t^\top Px_t$, where $P\succ 0$ is also called the quadratic kernel of the cost function associated with $K$ \cite{9691800}. Starting from \eqref{Cost} and using the principle of optimality, $P$ can be calculated by solving the model-based Bellman equation \cite{IJRNC}:
\begin{equation}\label{MBBE}
  P=Q+K^\top RK+(A+BK)^\top P(A+BK).
\end{equation} 

We introduce the following definition related to the gain $K$ and the quadratic kernel $P$:
\begin{Definition} [Stability of gain $K$ and kernel $P$]\label{Def1}
    A control gain $K$ is said to be stabilizing if $(A+BK)$ is Schur stable. A positive semi-definite matrix $P$ is said to be stabilizing if the gain $K=-(R+B^\top P B)^{-1} B^\top P A$ is stabilizing.
\end{Definition}
It is a well-known result \cite{lewis2012optimal} that the optimal controller solution to the LQR problem is a linear state-feedback, and the optimal feedback gain $K^*$ is obtained via:
\begin{subequations}\label{Kpolicyimprovement_EQ}
\begin{align}
K^*&=-(R+B^\top P^*B)^{-1}B^\top P^*A, \label{Kpolicyimprovement} \\
P^*&=Q+A^\top P^*A-A^\top P^*B(R+B^\top P^*B)^{-1}B^\top P^*A. \label{DARE} 
\end{align}
\end{subequations}
Here, $P^*$ is the quadratic kernel of the value function, i.e. of the cost associated with the optimal gain $K^*$, and is the unique solution of the discrete algebraic Riccati equation (DARE) \eqref{DARE}. The optimal gain $K^*$ is stabilizing. Therefore, based on Definition \ref{Def1}, $P^*$ is stabilizing.

Solving \eqref{DARE} directly is challenging, especially when dealing with a high number of system states. Value iteration and policy iteration offer an effective iterative approach to find the optimal gain $K^*$ and are introduced in the following subsections.

\subsection{Value Iteration (VI)}
The procedure of value iteration is summarized in Algorithm \ref{Algo1},  which requires knowledge of the system matrices $A$ and $B$ and only uses matrix multiplication to update the cost function. In the value iteration, the kernel $P_i$ is iteratively updated based on \eqref{DARE}, treating $P_i$ as the kernel of value function.

\vspace{1cm}

\begin{algorithm}[H]
  \caption{Value Iteration}\label{Algo1}
  \begin{algorithmic}
      \Require $A,B$, a positive semidefinite kernel $P_0$ 
      \For{$i=0,...,+\infty$}
        \State \textbf{Update the kernel $P_{i}$} 
        \State $P_{i+1}=Q+A^\top P_i A-A^\top P_i B (R+B^\top P_i B)^{-1} B^\top P_i A$
      \EndFor
  \end{algorithmic}
\end{algorithm}
The properties of Algorithm \ref{Algo1} are summarized in the following theorem.
\begin{Theorem}{Properties of VI \cite{SemiCone}} \label{theoremVIBasic}
 \\
If the system dynamics $(A,B)$ are stabilizable, then for all $P_0 \succeq 0$:
  \begin{enumerate}
    \item $\lim\limits_{i\rightarrow\infty}P_i=P^*$, thus the sequence $\{P_i\}$ converges asymptotically to $P^*$;
    \item if $P_0 \succeq P^*$, then $\lVert P_{i+1}-P^* \rVert_{P_\epsilon} \leq d \lVert P_{i}-P^* \rVert_{P_\epsilon}$ with $d\in(0,1)$, $\forall i \in \mathbb{Z}_+$. Thus, the sequence $\{P_i\}$ converges exponentially to $P^*$.
  \end{enumerate}
\end{Theorem}
\subsection{Policy Iteration (PI)}
The basic version of the policy iteration algorithm \cite{1099755} is summarized in Algorithm \ref{Algo2}. The PI algorithm is more complex than the VI algorithm, as the VI only requires the matrix multiplication and inversion, while PI involves solving the Lyapunov equation in the policy evaluation step. In the policy evaluation phase, the performance of $K_i$ is evaluated by using \eqref{MBBE}. In the policy improvement phase, the policy is improved by treating the evaluation $P_i$ as the kernel of value function and using \eqref{Kpolicyimprovement}. 
\begin{algorithm}[H]
  \caption{Policy Iteration}\label{Algo2}
  \begin{algorithmic}
      \Require $A,B$, a stabilizing policy gain $K_0$ 
      \For{$i=0,...,+\infty$}
        \State \textbf{Policy Evaluation: find $P_{i}$} 
        \State $P_{i}=Q+K_i^\top RK_{i}+(A+BK_{i})^\top P_{i}(A+BK_{i})$
        \State \textbf{Policy Improvement: update gain $K_{i+1}$}
        \State $K_{i+1}=-(R+B^\top P_iB)^{-1}B^\top P_iA$
      \EndFor
  \end{algorithmic}
\end{algorithm}
The properties of Algorithm \ref{Algo2} are summarized below.
\begin{Theorem}{Properties of PI \cite{1099755}\cite[Theorem 4]{IJRNC}} \label{theoremPIBasic}
If $K_0$ stabilizes \eqref{LTI}, then 
  \begin{enumerate}
    \item $P_0\succeq P_1 \succeq ... \succeq P^*$;
    \item $K_i$ stabilizes the system $(A,B)$, $\forall i \in \mathbb{Z}_+$;
    \item $\lim\limits_{i\rightarrow\infty}P_i=P^*$, $\lim\limits_{i\rightarrow\infty}K_i=K^*$;
    \item $\lVert P_{i+1}-P^* \rVert_F \leq c \lVert P_{i}-P^* \rVert_F$ with $c\in(0,1)$, $\forall i \in \mathbb{Z}_+$. Thus, the sequence $\{P_i\}$ converges exponentially to $P^*$.
  \end{enumerate}
\end{Theorem}
From Theorem \ref{theoremVIBasic}, the asymptotic convergence of VI is guaranteed by $P_0\succeq 0$, while the exponential convergence of VI can be only guaranteed under the condition $P_0 \succeq P^*$. From Theorem \ref{theoremPIBasic}, the exponential convergence of PI is guaranteed by initializing with a stabilizing policy gain $K_0$. 
\section{Convergence Analysis of VI and PI}\label{sec:C}
In this section, we relax the conditions for the exponential convergence properties of VI and PI algorithms. We begin by deriving the following lemma which combines the continuity of matrix eigenvalues \cite[Chapter 6]{meyer2023matrix} with the stability property of $P^*$ discussed earlier:

\begin{Lemma}[Stability of $P$ around $P^*$]\label{Lemma1}
    There exists a $\delta_0>0$ such that for any $P \in \mathcal{B}_{\delta_0}(P^*)$, $P$ is stabilizing.
\end{Lemma}

In the following two subsections, we investigate the exponential convergence properties of VI and PI algorithms within the region $\mathcal{B}_{\delta_0}(P^*)$. To facilitate our analysis, we introduce the following notations: 
\begin{subequations}\label{29102024}
\begin{align}
    L(P)&:=(R+B^\top P B)^{-1}B^\top P A,\\
    \mathcal{A}(P)&:=A-BL(P).
\end{align}
\end{subequations}
\subsection{Exponential Convergence of VI}
From Theorem \ref{theoremVIBasic}, the asymptotic convergence is guaranteed for any positive semidefinite $P_0$, and exponential convergence is achieved when $P_0\succeq P^*$. We now introduce new requirements for the exponential convergence of VI.

\begin{Theorem}[Local exponential convergence of VI]\label{LemmaVI}
    For any $P_i\in \mathcal{B}_{\delta_0}(P^*)$, with $\delta_0$ defined in Lemma \ref{Lemma1}, the following inequality holds:
    \begin{equation}\label{25102024}
    \lVert P_{i+1}-P^* \rVert_{P_\epsilon} \leq \alpha \lVert P_{i}-P^* \rVert_{P_\epsilon}, \quad \forall i \in \mathbb{Z}_+,
\end{equation}
where $\alpha \in (0,1)$ is a constant. Thus, the sequence $\{P_i\}$ converges exponentially to $P^*$ when $P_0\in \mathcal{B}_{\delta_0}(P^*)$.
\end{Theorem}
\begin{proof}
First, we define the Bellman operator $\mathcal{T}$ \cite{SemiCone} as follows:
\begin{equation}
    \mathcal{T}(P):=\mathcal{A}(P)^\top P \mathcal{A}(P) +L(P)^\top R L(P)+Q, 
\end{equation}
 which is a fixed point iteration in VI, i.e. $P_{i+1}=\mathcal{T}(P_i)$. Using this operator, a sequence $\{P_i\}$ is constructed, where $P_{i+1}=\mathcal{T}(P_i)$ with initialization $P_0$. 
Then the proof of Theorem \ref{LemmaVI} follows by establishing upper and lower bounds on the operator $\mathcal{T}(P)-P^*$, and then showing the conditions under which exponential convergence is guaranteed. 
    An upper bound of $\mathcal{T}(P)-P^*$ can be derived as:
     \begin{equation} \label{24102024} 
        \begin{split}
        &\mathcal{T}(P) - P^{*} \\
        &= 
        \begin{bmatrix}
        I  \\
        -L(P)
        \end{bmatrix}^{\top} 
        \underbrace{\begin{bmatrix}
        A^{\top}PA  & A^{\top}PB\\
        B^{\top}PA  & R+B^{\top}PB
        \end{bmatrix}}_{=:M(P)}
        \begin{bmatrix}
        I  \\
        -L(P)
        \end{bmatrix}   
        \\
        &\quad -
        \begin{bmatrix}
            I  \\
            -L(P^{*})
            \end{bmatrix}^{\top} 
           {M(P^{*})}
            \begin{bmatrix}
            I  \\
            -L(P^{*})
        \end{bmatrix}   \\
        &\preceq
        \begin{bmatrix}
            I  \\
            -L(P^{*})
        \end{bmatrix}^{\top} 
            (M(P)-M(P^*))
        \begin{bmatrix}
            I  \\
            -L(P^{*})
        \end{bmatrix}  \\
        &=\mathcal{A}(P^*)^{\top}(P-P^{*})\mathcal{A}(P^*),  
    \end{split}
     \end{equation}
    the inequality is due to the definition of $L(P)$ in \eqref{29102024} and \cite[Lemma 4]{lee2018primal}. Similarly, a lower bound can be derived by replacing $L(P^{*})$ with $L(P)$ at the first equality in \eqref{24102024} and using \cite[Lemma 4]{lee2018primal}:
    \begin{align*}
        \mathcal{T}(P) - P^{*} \succeq \mathcal{A}(P)^{\top}(P - P^{*})\mathcal{A}(P).
    \end{align*}
   From the upper bound and lower bound, we obtain the following for all $i\in \mathbb{Z}_+$:
    \begin{equation}\label{SW}
        \begin{aligned}
        \mathcal{A}(P_i)^{\top}(P_i - P^{*})\mathcal{A}(P_i)  \preceq \mathcal{T}(P_i) - P^{*}=\\
        P_{i+1}-P^*\preceq\mathcal{A}(P^*)^{\top}(P_{i} - P^{*})\mathcal{A}(P^*).\\  
        \end{aligned}
    \end{equation}
    Combining this with \cite[Lemma 6]{SemiCone}, we conclude:
    \begin{equation*}
    \begin{split}
        \lVert P_{i+1}-P^* \lVert_{P_\epsilon} \leq &\max\{\lVert \mathcal{A}(P_i) \rVert_{P_\epsilon}^2,\lVert \mathcal{A}(P^*)\rVert_{P_\epsilon}^2\}\\
        &\lVert P_{i}-P^* \lVert_{P_\epsilon}.
    \end{split}
    \end{equation*}
    By asymptotic convergence of $\{P_i\}$ from Theorem \ref{theoremVIBasic} and the definition of $\delta_0$, we have $\lim\limits_{i \rightarrow \infty}\mathcal{A}(P_i)=\mathcal{A}(P^*)$ and $\mathcal{A}(P_i)$ is Schur stable for all $i \in \mathbb{Z}_+$. Then we know $\max\{\lVert \mathcal{A}(P_i)\rVert_{P_\epsilon}^2,\lVert \mathcal{A}(P^*)\rVert_{P_\epsilon}^2\}< 1,~\forall i\in \mathbb{Z}_+$. We define $\alpha:=\sup\limits_i \{\lVert \mathcal{A}(P_i) \rVert_{P_\epsilon}^2\}$. Because $\lim\limits_{i \rightarrow \infty}\lVert \mathcal{A}(P_i) \rVert_{P_\epsilon}=\lVert \mathcal{A}(P^*) \rVert_{P_\epsilon} <1$ and $\lVert \mathcal{A}(P_i)\rVert_{P_\epsilon}< 1, \forall i \in \mathbb{Z}_+$, we have $\alpha \in (0,1)$. Then we conclude the proof of Theorem \ref{LemmaVI}. 
\end{proof} 

\begin{Remark}
    Unlike Theorem \ref{theoremVIBasic}, Theorem \ref{LemmaVI} does not require the condition $P_0 \succeq P^*$ for exponential convergence. Instead, exponential convergence is guaranteed when $P_0\in \mathcal{B}_{\delta_0}(P^*)$.  
\end{Remark}

Building on Theorem \ref{theoremVIBasic} and Theorem \ref{LemmaVI}, the following corollary establishes a larger region for local exponential convergence than what is currently available:
\begin{Corollary}[Exponential Convergence of VI]\label{theorem4}
        Defining set $\mathcal{S}:=\{P\succeq 0| P \succeq P^* \cup P\in \mathcal{B}_{\delta_0}(P^*)\}$ with $\delta_0$ defined in Lemma \ref{Lemma1}, for any $P_i \in \mathcal{S}$, we have:
    \begin{equation}
        \lVert P_{i+1}-P^* \rVert_{P_\epsilon} \leq v \lVert P_{i}-P^* \rVert_{P_\epsilon},\quad \forall i \in \mathbb{Z}_+,
    \end{equation}
    where $v \in (0,1)$ is a constant. Thus, the sequence $\{P_i\}$ converges exponentially to $P^*$ when $P_0 \in \mathcal{S}$. 
\end{Corollary}
The proof of Corollary \ref{theorem4} is a combination of Theorem \ref{theoremVIBasic} and Theorem \ref{LemmaVI}, where $v:=\max\{d,\alpha\}$. 

\subsection{Exponential Convergence of PI}
As outlined in Algorithm \ref{Algo2}, the PI procedure begins with an initial stabilizing control gain $K_0$ followed by the estimation of $P_0$ through the solution of a Lyapunov equation \eqref{MBBE} and continues by iterating $K_i$ and $P_i$.
To facilitate the comparison with the VI algorithm and the robustness analysis in Section \ref{Robust}, we consider the PI algorithm initialized with $P_0$ instead. For any $P_i\succeq P^*$, $K_{i+1}$ stabilizes the system $(A,B)$. Then from Theorem \ref{theoremPIBasic}, the sequence $\{P_i\}$ converges exponentially to $P^*$. 

Similarly to the analysis conducted for the VI algorithm, we investigate the convergence properties of the PI algorithm when $P_i \in \mathcal{B}_{\delta_0}(P^*)$.
\begin{Theorem}[Local exponential convergence of PI]\label{LemmaPI}
~
\begin{enumerate}
    \item For any $P_i\in \mathcal{B}_{\delta_0}(P^*)$, with $\delta_0$ defined in Lemma \ref{Lemma1}, the following inequality holds:
    \begin{equation}\label{251020241}
    \lVert P_{i+1}-P^* \rVert_F \leq \sigma_0 \lVert P_{i}-P^* \rVert_F, \quad \forall i \in \mathbb{Z}_{++},
\end{equation}
with $\sigma_0=c\in (0,1)$ defined in Theorem \ref{theoremPIBasic}. Thus, if $P_0\in \mathcal{B}_{\delta_0}(P^*)$, the sequence $\{P_i\}$ converges exponentially to $P^*$. The distance from $P^*$ decreases monotonically starting from $i=1$.
    \item There exists a constant $\delta_1 \in  (0,\delta_0]$, such that for any $P_i\in \mathcal{B}_{\delta_1}(P^*)$, the following inequality holds:
\begin{equation}
    \lVert P_{i+1}-P^* \rVert_F \leq \sigma_1 \lVert P_{i}-P^* \rVert_F,\quad  \forall i \in \mathbb{Z}_+,
\end{equation}
where $\sigma_1 \in (0,1)$. Thus, the sequence $\{P_i\}$ converges exponentially to $P^*$ when $P_0\in \mathcal{B}_{\delta_1}(P^*)$. 
\end{enumerate}
\end{Theorem}
The proof of Theorem \ref{LemmaPI} is provided in Appendix \ref{ProofofTheoremPI}.
\begin{Remark}\label{Remark2}
    In contrast to Theorem \ref{theoremPIBasic}, Theorem \ref{LemmaPI} does not require the condition $P_0 \succeq P^*$ for the exponential convergence. Instead, exponential convergence is guaranteed when $P_0\in \mathcal{B}_{\delta_0}(P^*)$. If $P_0 \in \mathcal{B}_{\delta_1}(P^*)$, the distance to $P^*$ decreases monotonically from the initial step $i=0$, which is necessary for the robustness analysis in Section \ref{Robust}. 
\end{Remark}

By combining Theorem \ref{theoremPIBasic} and Theorem \ref{LemmaPI}, we can derive the following corollary, which provides a larger region of initial conditions $P_0$ for which exponential convergence is guaranteed:
\begin{Corollary}[Exponential Convergence of PI]\label{theorem3}
     For any $P_i \in \mathcal{S}$, with $\mathcal{S}$ defined in Corollary \ref{theorem4}, we have:
        \begin{equation}
        \lVert P_{i+1}-P^* \rVert_F \leq c \lVert P_{i}-P^* \rVert_F,\quad \forall i \in \mathbb{Z}_{++},
    \end{equation}
    where $c \in (0,1)$ is defined in Theorem \ref{theoremPIBasic}. Thus, the sequence $\{P_i\}$ converges exponentially to $P^*$  when $P_0 \in \mathcal{S}$.
\end{Corollary}
The proof is a combination of Theorem \ref{theoremPIBasic} and Theorem \ref{LemmaPI}. We note that when $K_0$ is known a-priori to be stabilizing, the corresponding kernel $P_0$ automatically satisfies the condition $P_0 \succeq P^*$, and thus the enlargement of the exponential region provided by Theorem 4 is not of immediate use. However, it holds a significant value when the system matrices $(A,B)$ are unknown and thus the condition on $P_0$ required by Theorem \ref{theoremPIBasic} cannot be easily established a-prior. Establishing Theorem \ref{LemmaPI} is crucial for effectively analyzing the robustness of PI algorithm, which is a central theme of Section \ref{Robust}.

\subsection{Comparison between VI and PI}\label{CompareConvergence}
From the analysis in previous sections, we know that when $P_0\succeq P^*$, the sequences $\{P_i\}$ generated by both VI or PI algorithms converges exponentially to the optimal $P^*$, as graphically shown in the shaded region in Figure \ref{fig:enter-label}. In Theorem \ref{LemmaVI} and Theorem \ref{LemmaPI}, we identified the local region $\mathcal{B}_{\delta_0}(P^*)$ around $P^*$ of exponential convergence for both VI and PI. Additionally, we introduced $\mathcal{B}_{\delta_1}(P^*)$ specifically for PI, within which the distance from $P^*$ decreases monotonically for all $i \in \mathbb{Z}_+$, as illustrated in Figure \ref{fig:enter-label}. 
\begin{figure}[H]
    \centering
    \begin{tikzpicture}
    \draw[->] (0,0) -- (4,0) node[below] {$\lambda_1(P)$}; 
    \draw[->] (0,0) -- (0,4) node[left] {$\lambda_2(P)$};  
       \filldraw[gray!50, opacity=0.5] (2,2) rectangle (3.9,3.9);
    \draw[thick] (2,2) -- (2,4) ;
    \draw[thick] (2,2) -- (4,2) ;

       \draw[opacity=0.8, draw=black, thick] (2,2) circle [radius=1.5];
      \draw[opacity=0.8, draw=black, thick] (2,2) circle [radius=0.75];
    \filldraw [black] (2,2) circle (2pt) node[above right] {$P^*$};
    \node at (1.45,2.9) {\textcolor{black}{\fontsize{8}{14}\selectfont $\mathcal{B}_{\delta_0}(P^*)$}};
    \node at (2.0,1.7) {\textcolor{black}{\fontsize{8}{14}\selectfont $\mathcal{B}_{\delta_1}(P^*)$}};
    \node at (3,3.7) {\textcolor{black}{\fontsize{8}{14}\selectfont $P\succeq P^*$}};
\end{tikzpicture} 
    \caption{2-dimensional Graphic Representation}
    \label{fig:enter-label}
\end{figure}
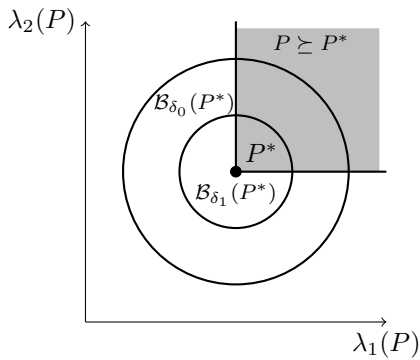
However, there is no explicit expression for $\delta_0$ (defined in Lemma \ref{Lemma1}). Theorem \ref{LemmaVI} and Theorem \ref{LemmaPI} only prove the existence of the region $\mathcal{B}_{\delta_0}(P^*)$ and $\mathcal{B}_{\delta_1}(P^*)$. Nonetheless, in certain special cases as the following one, we can provide sufficient conditions that ensure the stability of $P_0$, and we can use them to provide verifiable conditions on the initialization of VI and PI which ensure their exponential convergence beyond classic results from the literature.
\begin{Theorem}[Convergence of VI and PI with $P\succeq 0$]\label{theorem5}
    If the system matrix $A$ is Schur stable, then any $P \succeq 0$ is stabilizing. Therefore, for any $P_0 \succeq 0$, the sequences $\{P_i\}$ generated by both VI and PI converge exponentially to $P^*$.
\end{Theorem}
\begin{proof}
    Recalling the definition in \eqref{29102024}, we have that:
    \begin{equation*}
        \begin{split}
            \mathcal{A}(P)&= A - B(R + B^{\top}PB)^{-1}B^{\top}PA\\
            &= (BR^{-1}B^{\top}P + I)^{-1}A.
        \end{split}
    \end{equation*}
    Since $\lVert BR^{-1}B^{\top}P \rVert \geq  0$, because $BR^{-1}B^{\top}\succeq 0$ and $P\succeq 0$, it follows that: $\lVert (BR^{-1}B^{\top}P + I)^{-1} \rVert \leq  1$. Therefore, we have:
    \begin{equation*}
            \lVert A-BL(P) \rVert \leq  \lVert (BR^{-1}B^{\top}P+ I)^{-1}\rVert \lVert A \rVert < 1.
    \end{equation*}
    This implies that if $A$ is Schur stable, any positive semidefinite matrix $P$ is stabilizing. When $P_0 \succeq 0$, then $P_0$ is stabilizing, we can invoke Theorem \ref{LemmaVI} to establish that the sequence $\{P_i\}$ generated by the VI converges exponentially to the optimal value $P^*$. In the case of PI, we can conclude from Theorem \ref{LemmaPI} that the sequence $\{P_i\}$ converges to the optimal exponentially and the distance from $P^*$ decreases monotonically starting from $i=1$.  
\end{proof}
Theorem \ref{theorem5} enables the initialization of any $P_0 \succeq 0$ in both VI and PI algorithms when $\lVert A \rVert <1$. This theorem provides the exponential convergence guarantee for such an initialization for VI. For PI, convergence is also guaranteed, eliminating the need for an initializing stabilizing policy gain $K_0$ and instead allowing for initialization with any $P_0\succeq 0$.
\section{Robust Analysis of VI and PI}\label{Robust}
In the previous section, we analyzed the exponential convergence properties of VI and PI algorithms based on the assumption of perfect knowledge of the system matrices $(A,B)$. However, in practical scenarios, the system matrices are often unknown or only partially known. Therefore, we consider the case where, at each iteration $i$ of the algorithm, VI and PI employ estimates $\hat A_i$ and $\hat B_i$ in place of $A$ and $B$, respectively. We denote the differences between them as:
\begin{equation}
    \Delta A_i:= \hat{A}_i-A,~\Delta B_i:= \hat{B}_i-B.
\end{equation}
For the analysis in this section, we introduce two scalar sequences $\{a_i\}$ and $\{b_i\}$, whose entries are defined as:
\begin{equation}\label{251020242}
    a_i:=\lVert \Delta A_i\rVert_F,~ b_i:=\lVert \Delta B_i\rVert_F.
\end{equation}
This setting captures the case where a fixed model estimate is used ($\hat A_i=\hat A$ and $\hat B_i=\hat B$ for all $i\in \mathbb{Z}_+$) but also the more interesting case where a running estimate of the model is updated throughout the design process. The latter scenario arises for example in model-based RL \cite[Section 5]{annurev:/content/journals/10.1146/annurev-control-062922-090153}, and indirect data-driven control \cite{10383604}, where a system identification algorithm uses collected data to update online an estimate of the model. Given the use of VI and PI algorithms as building blocks of complex learning-based schemes \cite{IJRNC}, it is crucial to analyze their robustness in the face of inexact estimates of the system matrices.
\subsection{Robustness of Inexact VI}
The procedure of inexact VI algorithm formulated by estimate system matrices is given in Algorithm \ref{Algo3}. 
\begin{algorithm}[H]
  \caption{Value Iteration with Estimates $(\hat{A}_i,\hat{B}_i)$}\label{Algo3}
  \begin{algorithmic}
      \Require $\{\hat{A}_i\}$$\{\hat{B}_i\}$, a stabilizing $\hat{P}_0$ 
      \For{$i=0,...,+\infty$}\\
        $\hat{P}_{i+1}=\hat{A}_i^\top \hat{P}_i \hat{A}_i+Q-\hat{A}_i^\top \hat{P}_i \hat{B}_i (R+ \hat{B}_i^\top \hat{P}_i \hat{B}_i)^{-1} \hat{B}_i^\top \hat{P}_i \hat{A}_i$
      \EndFor
  \end{algorithmic}
\end{algorithm}
Note that the initial matrix $\hat{P}_0$ must be stabilizing for the true system. The following theorem analyzes the convergence properties of Algorithm \ref{Algo3}.
\begin{Theorem}[Robustness of VI]\label{TheoremRVI}
    Given $\alpha$ and $\delta_0$ as defined in Theorem \ref{LemmaVI} and Lemma \ref{Lemma1}, there always exist constants $\Bar{a}_{v}(\delta_0,\alpha)\geq 0$ and $\Bar{b}_{v}(\delta_0,\alpha)\geq 0$ such that if $\|a\|_\infty\leq \Bar{a}_{v}$, $\|b\|_\infty\leq \Bar{b}_{v}$ and $\hat{P}_0\in \mathcal{B}_{\delta_0}(P^*)$, where sequences $\{a_i\}$ and $\{b_i\}$ are defined in \eqref{251020242}, then:
    \begin{enumerate}
        \item $\hat{P}_i$ is stabilizing,  $\forall i \in \mathbb{Z}_+$;
        \item the following holds:
        \begin{equation}
        \begin{split}
            \|\hat{P}_{i} - P^*\|_{P_\epsilon} \leq \beta_{1}(\|\hat{P}_{0} - P^*\|_{P_\epsilon}, i) \\
            +\gamma_{1}(\|a\|_{\infty}) + \gamma_{2}(\|b\|_{\infty}),\quad \forall i \in \mathbb{Z}_+,
        \end{split}
        \end{equation}
        where $\beta_{1}(x, i) := \alpha^{i}x$; $\gamma_{1}(x) := \frac{\bar{v}_a}{1-\alpha}x$; $\gamma_{2}(x) := \frac{\bar{v}_b}{1-\alpha}x$ with constants $\bar{v}_a,~ \bar{v}_b >0$;
        \item if $\lim\limits_{i \to \infty} \|\Delta A_i\|_F = 0 $ and $\lim\limits_{i \to \infty} \|\Delta B_i\|_F = 0 $, then $\lim\limits_{i \to \infty} \|\hat{P}_{i} - P^*\|_{P_\epsilon} = 0 $.
    \end{enumerate}
\end{Theorem}
\begin{proof}
We prove each item in the theorem step by step:
\begin{enumerate}
    \item Matrix $(R + \hat{B}_i^{\top}\hat{P}_i\hat{B}_i)$ is always invertible because of $R\succ 0$. For $\hat{P}_i \in \mathcal{B}_{\delta_0}(P^*)$, $\hat{P}_i$ is stabilizing and ${\mathcal{A}}(\hat{P}_{i})$ is Schur stable. Defining $$\hat{\mathcal{A}}(\hat{P}_{i}) := \hat{A}_{i+1} - \hat{B}_{i+1}\hat{L}(\hat{P}_{i}),$$ 
$$\hat{L}(\hat{P}_{i}):=(R+\hat{B}_{i+1}^\top \hat{P}_{i} \hat{B}_{i+1})^{-1}\hat{B}_{i+1}^\top \hat{P}_{i} \hat{A}_{i+1},$$
we know that $\hat{\mathcal{A}}(\hat{P}_{i})$ is a continuous function of $\hat{A}_{i+1}$ and $\hat{B}_{i+1}$. By continuity there exist constants $d_a>0$ and $d_b>0$ such that $\hat{\mathcal{A}}(\hat{P}_{i})$ is Schur stable when $\|\Delta A_{i+1}\|_F \leq d_a$, $\|\Delta B_{i+1}\|_F \leq d_b$ for all $i\in \mathbb{Z}_+$.
Then we derive the relation for the VI algorithm based on the estimates system matrices where the invertibility of the operator $\mathcal{L}^{-1}$ defined in \eqref{PIrel} is guaranteed by Schur stable matrices ${\mathcal{A}}(\hat{P}_{i})$ and ${\hat{\mathcal{A}}}(\hat{P}_{i})$:
\begin{equation}
    \begin{aligned}
        \hat{P}_{i+1}& = \mathcal{L}^{-1}_{\hat{\mathcal{A}}(\hat{P}_{i})}(-\hat{L}(\hat{P}_{i})^{\top}R\hat{L}(\hat{P}_{i}) - Q - \hat{E}_{i+1}) \\
            &= \mathcal{L}^{-1}_{\mathcal{A}(\hat{P}_{i})}(-{L}(\hat{P}_{i})^{\top}R{L}(\hat{P}_{i}) - Q - E_{i+1}) \\
            &\quad+\mathcal{E}_{\textbf{VI}}(\Delta A_i, \Delta B_i),
            \label{eq:simplified_relation_Piplus1_VI_hat}   
    \end{aligned}
\end{equation}
where
\begin{align*} 
    &E_{i+1} := \mathcal{A}(\hat{P}_{i})^{\top} (\hat{P}_{i} - \hat{P}_{i+1}) \mathcal{A}(\hat{P}_{i}) ; \\
   & \hat{E}_{i+1} := \hat{\mathcal{A}}(\hat{P}_{i})^{\top} (\hat{P}_{i} - \hat{P}_{i+1}) \hat{\mathcal{A}}(\hat{P}_{i}); 
\end{align*}
    we denote $\mathcal{E}_{\textbf{VI}}(\Delta A_i,  \Delta B_i)$ as the difference between the $i$-th iteration step of true system and estimate system:
\begin{equation}
    \begin{aligned}
    &\mathcal{E}_{\textbf{VI}}(\Delta A_i,  \Delta B_i) := \\
    &\quad \mathcal{L}^{-1}_{\hat{\mathcal{A}}(\hat{P}_{i})}(-\hat{L}(\hat{P}_{i})^{\top} R \hat{L}(\hat{P}_{i}) - Q - \hat{E}_{i+1}) \\
    &\quad - \mathcal{L}^{-1}_{\mathcal{A}(\hat{P}_{i})}(-L(\hat{P}_{i})^{\top} R L(\hat{P}_{i}) - Q - E_{i+1}).  \label{eq:error_Piplus1_hat_VI}   
    \end{aligned}
\end{equation}
    Using the inequality in \cite[Equation 19]{RobustPI}, and similar to the proof of \cite[Lemma 5]{RobustPI}, it can be verified that:
    \begin{equation}
        \begin{aligned}
        \|\Delta \mathcal{P}^{\textbf{VI}}_i\|_F&:= \|\mathcal{P}(\hat{\mathcal{A}}(\hat{P}_{i})) - \mathcal{P}(\mathcal{A}(\hat{P}_{i}))\|_F \\ &\leq a_1 \|\Delta A_i\|_F+a_2 \|\Delta B_i\|_F,\\
        \end{aligned}
    \end{equation}
        \begin{equation}
        \begin{aligned}
        \|\Delta \mathcal{V}^{\textbf{VI}}_i\|_F:&= \|vec(L(\hat{P}_{i}) R L(\hat{P}_{i}) - \hat{L}(\hat{P}_{i})^{\top}R\hat{L}(\hat{P}_{i})) \\
        &+ vec(E_{i+1} - \hat{E}_{i+1})\|_F\\
        & \leq a_3\|\Delta A_i\|_F+a_4\|\Delta B_i\|_F,\\
        \end{aligned}
    \end{equation}
    where $a_1, a_2, a_3, a_4$ are polynomials of $(A, B, Q, R)$ and they can be computed using matrix multiplication. For the detailed computation steps and derivation of these polynomials, we refer to \cite[Appendix D6]{IJRNC}.
    Using these results, combined with the inequality in \cite[Equation 19]{RobustPI}, we can obtain:
    \begin{equation}
        \begin{aligned}
        &\|\mathcal{E}_{\textbf{VI}}(\Delta A_i, \Delta B_i)\|_F \\
        &= \|\mathcal{P}(\hat{\mathcal{A}}(\hat{P}_{i}))^{-1}vec(-\hat{L}(\hat{P}_{i})^{\top}R\hat{L}(\hat{P}_{i}) - Q - \hat{E}_{i+1})    \\ 
        &\quad - \mathcal{P}(\mathcal{A}(\hat{P}_{i}))^{-1}vec(-L(\hat{P}_{i})^{\top}RL(\hat{P}_{i}) - Q - E_{i+1})\|  \\
        &\leq \|\mathcal{P}(\hat{\mathcal{A}}(\hat{P}_{i}))^{-1}\|_F(\|\Delta \mathcal{V}^{\textbf{VI}}_i\|_F + \|\mathcal{P}(\mathcal{A}(\hat{P}_{i}))^{-1}\|_F  \times \\
        &\quad \|vec(-L(\hat{P}_{i})^{\top}RL(\hat{P}_{i}) - Q - E_{i+1})\|\|\Delta \mathcal{P}^{\textbf{VI}}_i\|_F) \\
        &\leq \bar{v}_a\|\Delta A_i\|_F+ \bar{v}_b\|\Delta B_i\|_F,   
        \end{aligned}
    \end{equation}
    where $\bar{v}_a,\bar{v}_b$ are polynomials of $(A, B, Q, R)$. Combining this with \cite[Lemma 6]{SemiCone}, it can be verified that:
    \begin{equation}
        \begin{aligned}
        &\|\mathcal{E}_{\textbf{VI}}(\Delta A_i, \Delta B_i)\|_{P_\epsilon}\leq \|\mathcal{E}_{\textbf{VI}}(\Delta A_i, \Delta B_i)\|_F \\
        &\leq \bar{v}_a\|\Delta A_i\|_F+ \bar{v}_b\|\Delta B_i\|_F \leq \bar{v}_a \bar{a}_v+ \bar{v}_b\bar{b}_v=:\epsilon_1.\label{eq:errorrubustVI}
        \end{aligned}
    \end{equation}
 Next, we show that if $\hat{P}_i \in \mathcal{B}_{\delta_0}(P^*)$ and $\epsilon_1 = (1-\alpha)\delta_0$, then $\hat{P}_{i+1} \in \mathcal{B}_{\delta_0}(P^*)$. According to Theorem \ref{LemmaVI}, if $P_i \in \mathcal{B}_{\delta_0}(P^*)$, then: 
    \begin{equation}
        \begin{aligned}
        &\|\hat{P}_{i+1} - P^*\|_{P_{\epsilon}} \\
        &=\|(\mathcal{L}^{-1}_{\mathcal{A}(\hat{P}_{i})}(-L(\hat{P}_{i})^{\top}RL(\hat{P}_{i}) - Q - E_{i+1}) - P^*) \\
        &+ \mathcal{E}_{\textbf{VI}}(\Delta A_i, \Delta B_i)\|_{P_{\epsilon}}  \\
        &\leq \|\mathcal{L}^{-1}_{\mathcal{A}(\hat{P}_{i})}(-L(\hat{P}_{i})^{\top}RL(\hat{P}_{i}) - Q - E_{i+1}) - P^*\|_{P_{\epsilon}} \\
        &+ \|\mathcal{E}_{\textbf{VI}}(\Delta A_i,\Delta B_i)\|_{P_{\epsilon}}    \\
        &\leq \alpha \|\hat{P}_{i} - P^*\|_{P_{\epsilon}} + \|\mathcal{E}_{\textbf{VI}}(\Delta A_i, \Delta B_i)\|_F \\
        &\leq \alpha \delta_0 + \epsilon_1 \leq \delta_0   \label{eq:rubustISSVI}
        \end{aligned}
    \end{equation}
    Then we can select $\epsilon_1 = (1-\alpha)\delta_0$ and get its associated $\Bar{a}_{v}:=\min\left\{\frac{(1-\alpha)\delta_0}{2\bar{v}_a},d_a\right\}$, $\Bar{b}_{v}:=\min\left\{\frac{(1-\alpha)\delta_0}{2\bar{v}_b},d_b\right\}$, such that $\hat{P}_i \in \mathcal{B}_{\delta_0}(P^*),~\forall i \in \mathbb{Z}_+$. This completes the first part.
    \item Furthermore, we have:
    \begin{equation}
        \begin{aligned}
        &\|\hat{P}_{i} - P^*\|_{P_{\epsilon}} \leq \alpha \|\hat{P}_{i-1} - P^*\|_{P_{\epsilon}} + \|\mathcal{E}_{\textbf{VI}}(\Delta A_i, \Delta B_i)\|_{P_{\epsilon}}   \\
        &\leq \alpha \|\hat{P}_{i-1} - P^*\|_{P_{\epsilon}} + \bar{v}_a\|\Delta A_i\|_F+ \bar{v}_b\|\Delta B_i\|_F   \\
        &\leq \alpha \|\hat{P}_{i-1} - P^*\|_{P_{\epsilon}} + \bar{v}_a\|a_v\|_\infty+ \bar{v}_b\|b_v\|_\infty   \\
        &\leq \alpha^2 \|\hat{P}_{i-2} - P^*\|_{P_{\epsilon}} + (1 + \alpha)(\bar{v}_a\|a_v\|_\infty+ \bar{v}_b\|b_v\|_\infty )  \\
    &\leq \alpha^i \|\hat{P}_{0} - P^*\|_{P_{\epsilon}} \\
    &+ (1 + \alpha + \dots + \alpha^{i-1})(\bar{v}_a\|a_v\|_\infty+ \bar{v}_b\|b_v\|_\infty ) \\
        &< \alpha^i \|\hat{P}_{0} - P^*\|_{P_{\epsilon}} + \frac{\bar{v}_a}{1-\alpha}\| a_v\|_{\infty}  + \frac{\bar{v}_b}{1-\alpha}\| b_v \|_{\infty} .   
        \end{aligned}
    \end{equation}
    This completes the second part.
    \item The third part is a standard corollary of input-to-state stability results and it can be proved for example by following the steps outlined in \cite[Appendix D3]{IJRNC}. 
\end{enumerate}
\end{proof}
\subsection{Robustness of Inexact PI}
The procedure for the inexact policy iteration algorithm is outlined in Algorithm \ref{Algo4}. 
\begin{algorithm}[H]
  \caption{Policy Iteration with Estimates $(\hat{A}_i,\hat{B}_i)$}\label{Algo4}
  \begin{algorithmic}
      \Require $\{\hat{A}_i\}$$\{\hat{B}_i\}$, a stabilizing gain $\hat{K}_0$ 
      \For{$i=0,...,+\infty$}
        \State $\hat{P}_{i}=Q+\hat{K}_i^\top R\hat{K}_{i}+(\hat{A}_i+\hat{B}_i\hat{K}_{i})^\top \hat{P}_{i}(\hat{A}_i+\hat{B}_i\hat{K}_{i})$
        \State $\hat{K}_{i+1}=-(R+\hat{B}_{i+1}^\top \hat{P}_i\hat{B}_{i+1})^{-1}\hat{B}_{i+1}^\top \hat{P}_i\hat{A}_{i+1}$
      \EndFor
  \end{algorithmic}
\end{algorithm}
The following theorem analyzes the convergence properties of Algorithm \ref{Algo4}.
\begin{Theorem}[Robustness of PI]\label{TheoremRPI}
    Given $\sigma_1$ and $\delta_1$ defined in Theorem \ref{LemmaPI}, there always exist constants $\Bar{a}_{p}(\delta_1,\sigma_1)\geq 0$ and $\Bar{b}_{p}(\delta_1,\sigma_1)\geq 0$ such that if $\|a\|_\infty\leq \Bar{a}_{p}$, $\|b\|_\infty\leq \Bar{b}_{p}$ and $\hat{P_0}\in \mathcal{B}_{\delta_1}(P^*)$, where sequences $\{a_i\}$ and $\{b_i\}$ are defined in \eqref{251020242}, then:
    \begin{enumerate}
        \item $\hat{K}_i$ is stabilizing, $\forall i \in \mathbb{Z}_{+}$;
        \item the following holds:
        \begin{equation}
        \begin{split}
            \|\hat{P}_{i} - P^*\|_F \leq \beta_{2}(\|\hat{P}_{0} - P^*\|_F, i) \\
            +\gamma_{3}(\|a\|_{\infty}) + \gamma_{4}(\|b\|_{\infty}),\quad \forall i \in \mathbb{Z}_{+},
        \end{split}
        \end{equation}
        where $\beta_{2}(x, i) := \sigma_1^{i}x$; $\gamma_{3}(x) := \frac{\bar{p}_a}{1-\sigma_1}x$; $\gamma_{4}(x) := \frac{\bar{p}_b}{1-\sigma_1}x$ with constants $\bar{p}_a, \bar{p}_b >0$;
        \item if $\lim\limits_{i \to \infty} \|\Delta A_i\|_F = 0 $ and $\lim\limits_{i \to \infty} \|\Delta B_i\|_F = 0 $, then $\lim\limits_{i \to \infty} \|\hat{P}_{i} - P^*\|_F = 0 $.
    \end{enumerate}
\end{Theorem}

The proof of this result follows similar arguments to that of Theorem \ref{TheoremRVI} and can be found in the Appendix \ref{ProofRobustnessPI}.

   Theorem \ref{TheoremRVI} and Theorem \ref{TheoremRPI} show that both VI and PI algorithms have an inherent robustness against uncertainties in the system matrices, when the uncertainties remain within the bounds specified in the theorems. 

\section{Simulation}\label{sec:S}
In this section, we present some numerical results\footnote{The Matlab codes used to generate these results are accessible from the repository: \url{https://github.com/col-tasas/2024-ConvergenceRobustness-VIPI}} to compare the convergence and robustness of VI and PI algorithms. We consider the following system which was already used in prior studies \cite{IJRNC,9691800}:
\begin{equation}\label{LTIsimulation}
  x_{t+1}=\underbrace{\left[\begin{array}{ccc}
            1.01 & 0.01 & 0 \\
            0.01 & 1.01 & 0.01 \\
            0 & 0.01 & 1.01 
          \end{array}\right]}_A x_t+\underbrace{\left[\begin{array}{ccc}
            1 & 0 & 0 \\
            0 & 1 & 0 \\
            0 & 0 & 1 
          \end{array}\right]}_B u_t.
\end{equation}
The weight matrices $Q$ and $R$ are set to $0.001I_3$ and $I_3$, respectively.

Figure \ref{fig:convergence} illustrates the convergence properties of the VI and PI algorithms assuming perfect knowledge of the system matrices $(A,B)$.
\begin{figure}[H]
    \centering
    \includegraphics[width=1\linewidth]{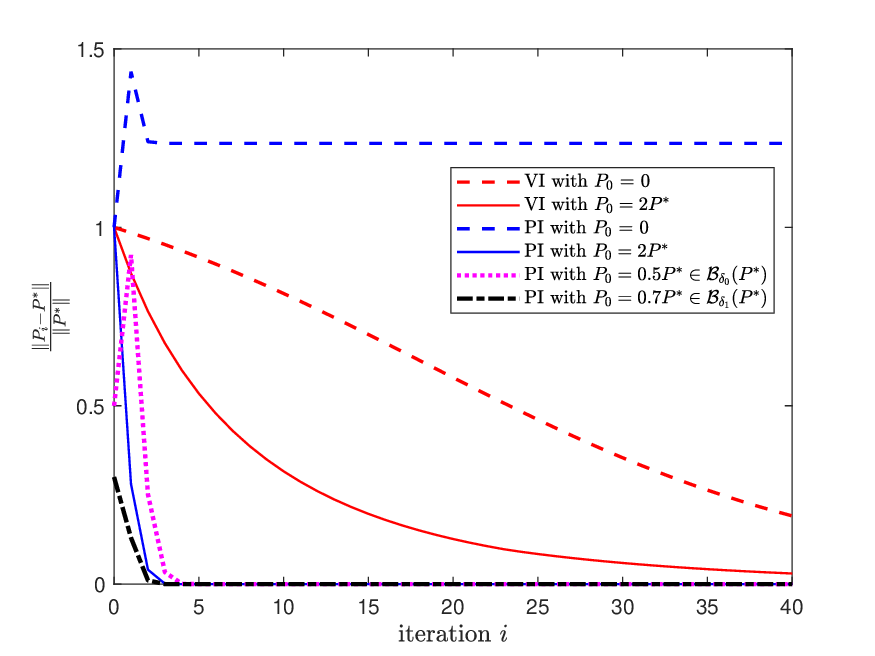}
    \caption{Convergence of VI and PI}
    \label{fig:convergence}
\end{figure}

The figure above presents six curves illustrating the convergence behavior of the VI and PI algorithms under different initializations. The blue and red solid lines depict the convergence of the VI and PI algorithms, respectively, with initial condition $P_0=2P^* \succeq P^*$, corresponding to Theorem \ref{theoremVIBasic} and Theorem \ref{theoremPIBasic}. When initialized with $P_0=0$, VI converges to the optimal solution (red dashed line), consistent with Theorem \ref{LemmaVI}, whereas PI does not (blue dashed line). In the case of a closer initialization ($P_0=0.5P^*$) to $P^*$ (magenta dotted line), the sequence $\{P_i\}$ converges to the optimal, and the distance between $P_i$ and $P^*$ decreases monotonicity after the first step, as described in item 1 of Theorem \ref{LemmaPI}. Finally, when the initialization ($P_0=0.7P^*$) is even closer to $P^*$, PI converges monotonically to the optimal solution as shown by the black dash-dotted line, in alignment with item 2 in Theorem \ref{LemmaPI}.

Next, we investigate the robustness properties of the VI and PI algorithms numerically.
\begin{figure}[H]
    \centering
    \includegraphics[width=1\linewidth]{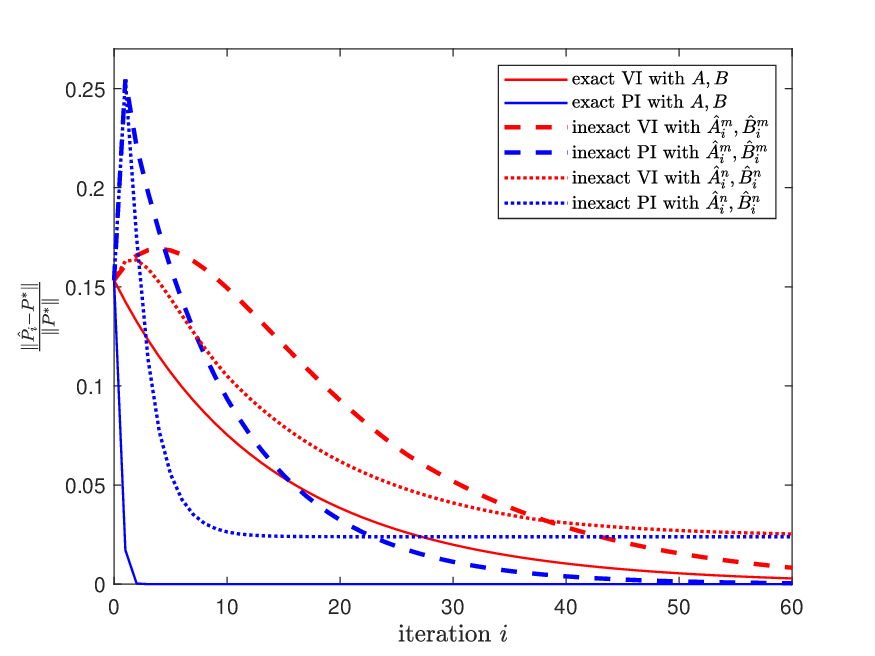}
    \caption{Robustness of VI and PI}
    \label{fig:robustness}
\end{figure}
We consider two different scenarios for the estimate system matrices used by the algorithms. In the first case, we have:
$$\hat{A}^m_i=A+0.9^i\times0.01I,~\hat{B}^m_i=B+0.9^i\times 0.01I,$$ which satisfy the conditions $\lim\limits_{i\rightarrow \infty}\hat{A}_i=A$ and $\lim\limits_{i\rightarrow \infty}\hat{B}_i=B$. From Figure \ref{fig:robustness}, it is evident that $\{\hat{P}_i\}$ converges to the optimal for both VI and PI, as established in item 3 of Theorem \ref{TheoremRVI} and Theorem \ref{TheoremRPI}. In the second case, we have
$$\hat{A}^n_i=A+(0.6^i+0.1)\times0.01I,$$ $$~\hat{B}^n_i=B+(0.6^i+0.1)\times 0.01I,$$
As expected, the algorithms converge but do not recover the optimal kernel matrix $P^*$ because of the non-vanishing mismatch in the estimate matrices.

\section{Conclusion}\label{conclusion}
This study contributes a thorough analysis of the convergence and robustness properties of value and policy iteration algorithms within the framework of the linear quadratic regulator problem. We extend the conditions for the exponential convergence of both VI and PI algorithms, which is provided in \cite{SemiCone} and, building on them, present input-to-state stability results to evaluate the robustness of the VI and PI algorithms against uncertainties in the system matrices. Additionally, we provide numerical examples to illustrate our analytical findings. In future work, we aim to integrate the robustness analysis with online system identification using noisy data to assess the performance of indirect data-driven VI and PI algorithms. 

\appendix
\subsection{Proof of Theorem \ref{LemmaPI}}\label{ProofofTheoremPI}
\begin{proof}
First, we prove the first argument in Theorem \ref{LemmaPI}. If $P_i\in \mathcal{B}_{\delta_0}(P^*)$, $K_{i+1}$ stabilizes the system as indicated by Lemma \ref{Lemma1}. It follows that $P_{i+1}\succeq P^*$, allowing us to apply Theorem \ref{theoremPIBasic} to ensure exponential convergence. Thus, we conclude that if $P_0\in \mathcal{B}_{\delta_0}(P^*)$, $\{P_i\}$ convergences to $P^*$ exponentially, with the exception of the transition from $P_0$ to $P_1$.

Now, we move to prove the second argument. For analysis purposes, we can construct a sequence ${P_i}$ by incorporating the policy improvement step within the policy evaluation step.This allows writing the evolution of $P_i$ compactly as follows:
\begin{equation}\label{relationPseq}
\begin{aligned}
   \mathcal{L}_{\mathcal{A}(P_i)} (P_{i+1})= - L(P_i)^\top R L(P_i)-Q,
\end{aligned}
\end{equation}
where $\mathcal{L}_X{(Y)}:=X^\top Y X-Y$. Following the approach outlined in \cite[Equation 9]{RobustPI}, we can verify the following relationship:
\begin{equation}\label{11}
    vec(\mathcal{L}_X(Y))=\mathcal{P}(X)vec(Y),
\end{equation}
where $\mathcal{P}(X):=X^\top \otimes X^\top-I\otimes I$. By applying this relationship to rewrite \eqref{relationPseq}, we obtain:
\begin{equation}\label{12}
    \mathcal{P}(\mathcal{A}(P_i))vec(P_{i+1})=vec(- L(P_i)^\top R L(P_i)-Q).
\end{equation}
This formulation allows us to analyze the convergence properties of the policy iteration algorithm by leveraging a discrete-time dynamical systems interpretation on the sequence $\{P_i\}$ \cite[Theorem 6]{IJRNC}. 
From Lemma \ref{Lemma1}, for any $P_i\in \mathcal{B}_{\delta_0}(P^*)$, the eigenvalues of $\mathcal{A}(P_i)$ lie in the interval $ (-1,1)$. Consequently, from the definition of $P(X)$, we know that the eigenvalues of $P(\mathcal{A}(P_i))$ lie in $ (-2,0)$, ensuring that $P^{-1}(\mathcal{A}(P_i))$ always exists.
    From \eqref{12}, we have: $$vec(P_{i+1})=\mathcal{P}^{-1}(\mathcal{A}(P_i))vec(- L(P_i)^\top R L(P_i)-Q),$$ which allows us to recursively compute the vectorized $P_{i+1}$ from $P_i$. Thus, we denote $\mathcal{L}^{-1}_{(\cdot)}{(\cdot)}$ as the inverse operator defined in \eqref{relationPseq}, which describes the relation between $P_{i+1}$ and $P_{i}$, as:
    \begin{equation}\label{PIrel}
        P_{i+1}=\mathcal{L}^{-1}_{(\mathcal{A}(P_i))}(-L(P_i)^\top R L(P_i)-Q).
    \end{equation}
     Adding the term $(- L(P_{i})^{\top}B^{\top}P^{*}A - A^{\top}P^{*}BL(P_{i}) + L(P_{i})^{\top}B^{\top}P^{*}BL(P_{i}))$ to both sides of DARE \eqref{DARE} yields:
    \begin{equation}\label{13}
    \begin{split}
       & \mathcal{L}_{(\mathcal{A}(P_i))}(P^{*}) = L(P^{*})^\top(R + B^{\top}P^{*}B)L(P^{*}) \\
       &-L(P_{i})^{\top}B^{\top}P^{*}A  - A^{\top}P^{*}BL(P_{i})\\
      & + L(P_{i})^{\top}(R+B^{\top}P^{*}B)L(P_{i}) - \mathcal{L}_{(\mathcal{A}(P_i))}(P_{i+1}) .    
    \end{split}
    \end{equation}
    From \eqref{13}, we obtain:
        \begin{equation}\label{14}
        \begin{split}
            \mathcal{L}_{(\mathcal{A}(P_i))}(P^{*} - P_{i+1}) &=\\
            (L(P^{*}) - L(P_{i}))^{\top}(R+B^{\top}&P^{*}B)(L(P^{*}) - L(P_{i})).
        \end{split}
    \end{equation}
    With the invertibility of $P(\mathcal{A}(P_i))$, we can rewrite \eqref{14} as:
    \begin{equation}\label{eq:Vectorization}
    \begin{split}
        vec(P^{*} - P_{i+1}) &=\mathcal{P}(\mathcal{A}(P_i))^{-1}\\
        vec((L(P^*) - L(P_{i}))^{\top}&(R+B^{\top}P^{*}B)(L(P^*) - L(P_{i})) ).
    \end{split}
    \end{equation}
    Taking the vector 2-norm on both sides of the equation above, we have:
    \begin{equation}
        \begin{aligned}
       & \|P^{*} - P_{i+1}\|_F \leq \|\mathcal{P}(\mathcal{A}(P_{i}))^{-1}\|_{F}\\
            & \|(R+B^{\top}P^{*}B)\|_{F}\|L(P^*) - L(P_{i})\|^2_F. \label{eq:difference_Pstar_Piplus1_norm_invec}
        \end{aligned}
    \end{equation}
    Based on the definition of $L(P)$, we derive the following:
    \begin{equation}
            \begin{aligned}
        &\|L(P^*) - L(P_{i})\|_{F}^{2} \leq \|R ^{-1}\|_F^2\|B\|_{F}^{2}\|A\|_{F}^{2} \\
        &  (1 +\|R^{-1}\|_{F}\|B\|^2_{F} \|P_i\|_{F})^{2}\|P^{*} - P_{i}\|_{F}^{2}. \label{eq:relation_Fnorm_Pnorm_squar}
        \end{aligned}
    \end{equation}
Substituting \eqref{eq:relation_Fnorm_Pnorm_squar} into \eqref{eq:difference_Pstar_Piplus1_norm_invec}, we obtain:
    \begin{equation}
        \begin{aligned}
            &\|P^{*} - P_{i+1}\|_F \leq a_0(P^{*})a_1(P_i)\|(P^{*} - P_{i})\|_{F}^{2},\\   
        \end{aligned}
    \end{equation}
    where $a_0(P^{*}):=\|(R+B^{\top}P^{*}B)\|_{F}\|R^{-1}\|^2_F\|B\|_{F}^{2}\|A\|_{F}^{2}$ and $a_1(P_i):=\|\mathcal{P}(\mathcal{A}(P_i))^{-1}\|_{F}(1 + \|R^{-1}\|\|B\|^2_{F}  \|P_i\|_{F})^{2}$. There exists a constant $c_1 > 0$, such that $a_1(P_i) \leq c_1, \forall P_i \in {\mathcal{B}}_{\delta_0}(P^*)$. Setting $c_0:=c_1a_0(P^{*})$, we can choose a sufficiently small $\delta_1 \in (0, \delta_0]$ such that $c_{0}\delta_{1} \leq 1$. Then, for any $P_i \in \mathcal{B}_{\delta_1}(P^*) \subseteq \mathcal{B}_{\delta_0}(P^*)$, we have $\|P^* - P_i\|_F < \delta_1$. Consequently, $\sigma_1 := c_0 \|P^* - P_i\|_F < c_0 \delta_1 \leq 1$. This completes the proof of Theorem \ref{LemmaPI}.
\end{proof}
\subsection{Proof of Theorem \ref{TheoremRPI}}\label{ProofRobustnessPI}

\begin{proof}
  Similar to the proofs of Theorem \ref{TheoremRVI} and \cite[Lemma 4]{RobustPI}, there exist constants $d_a>0$ and $d_b>0$ such that $\hat{\mathcal{A}}(\hat{P}_{i})$ is invertible and $\hat{K}_i$ is stabilizing, when $\|\Delta A_i\|_F \leq d_a$, $\|\Delta B_i\|_F \leq d_b$ and $\hat{P}_{i} \in \mathcal{B}_{\delta_1}(P^*)$ for all $i$. we can derive the relationship for inexact PI using \eqref{relationPseq}:
\begin{equation}
    \begin{aligned}
        \hat{P}_{i+1} = &\mathcal{L}^{-1}_{\mathcal{A}(\hat{P}_{i})}(-L(\hat{P}_{i})^{\top}RL(\hat{P}_{i}) - Q) + \mathcal{E}_{\textbf{PI}}(\Delta A_i, \Delta B_i),   
        \label{eq:simplified_relation_Piplus1_Pi_hat}   
    \end{aligned}
\end{equation}
where  
\begin{equation*}
    \begin{split}
             \mathcal{E}_{\textbf{PI}}(\Delta A_i, \Delta B_i)&:= \mathcal{L}^{-1}_{\hat{\mathcal{A}}(\hat{P}_{i})}(-\hat{L}(\hat{P}_{i})^{\top}R\hat{L}(\hat{P}_{i}) - Q) \\
            &-\mathcal{L}^{-1}_{\mathcal{A}(\hat{P}_{i})}(-{L}(\hat{P}_{i})^{\top}R{L}(\hat{P}_{i}) - Q ).
    \end{split}
\end{equation*}
Using the same techniques as in the proof of Theorem \ref{TheoremRVI} and combining with the inequality in \cite[Appendix]{RobustPI}, we obtain:
    \begin{equation}
        \begin{aligned}
        \|\mathcal{E}_{\textbf{PI}}(\Delta A_i, \Delta B_i)\|_F &\leq \bar{p}_a\|\Delta A_i\|_F+ \bar{p}_b\|\Delta B_i\|_F \\
        \leq \bar{p}_a \bar{a}_p+ \bar{p}_b\bar{b}_p&:=\epsilon_2.
        \label{eq:errorrubustPI}
        \end{aligned}
    \end{equation}
    We can also show that if $\hat{P}_i \in \mathcal{B}_{\delta_1}(P^*)$ and $\epsilon_2=(1-\sigma_1)\delta_1$, then $\hat{P}_{i+1} \in \mathcal{B}_{\delta_1}(P^*)$ holds, then we finish the proof of the first item.
    Thus, we can define $\Bar{a}_{p}:=\min\left\{\frac{(1-\sigma_1)\delta_1}{2\bar{p}_a},d_a\right\}$, $\Bar{b}_{p}:=\min\left\{\frac{(1-\alpha)\delta_1}{2\bar{p}_b},d_b\right\}$. Similar to the case of VI, we can prove 2) and 3) for PI.
\end{proof}
\bibliographystyle{unsrt}
\bibliography{ECC}
	
\end{document}